\documentclass[11pt]{article}
\usepackage[utf8]{inputenc}   % arXiv pdflatex-compatible encoding
\usepackage[T1]{fontenc}
\usepackage{lmodern}          % vector Latin Modern fonts
\usepackage[a4paper,margin=1in]{geometry}
\usepackage{microtype}
\usepackage{setspace}         % optional: \onehalfspacing for review
\usepackage{titling}

\usepackage{amsmath,amssymb,amsthm,mathtools}
\usepackage{bm}
\usepackage{physics}          % optional: \dv, \pdv, etc. Remove if unused.

\usepackage{graphicx}
\graphicspath{{figs/}}        % put figures under ./figs
\usepackage{caption}
\usepackage{subcaption}
\usepackage{booktabs}
\usepackage{siunitx}
\usepackage{algorithm}
\usepackage[noend]{algpseudocode}

\usepackage[numbers,sort&compress]{natbib} % -> use plainnat or unsrtnat
\usepackage{hyperref}
\hypersetup{
  colorlinks=true,
  linkcolor=blue,
  citecolor=blue,
  urlcolor=blue,
  pdftitle={Similarity as Thermodynamic Work: Between Depth and Diversity --- from Information Distance to Ugly Duckling},
  pdfauthor={Kentaro Imafuku}
}
\usepackage[nameinlink,capitalize,noabbrev]{cleveref} % must be loaded after hyperref

\newtheorem{proposition}{Proposition}

\theoremstyle{definition}
\newtheorem{definition}{Definition}
\theoremstyle{remark}

\usepackage{authblk}
\title{\textbf{Similarity as Thermodynamic Work: \\Between Depth and Diversity \\--- from Information Distance to Ugly Duckling}}
\author{Kentaro IMAFUKU}
\affil[]{National Institute of Advanced Industrial Science and Technology (AIST)}
\date{\empty}
\begin{document}
\maketitle

% ----------------------------
% Abstract
% ----------------------------
\begin{abstract}
Defining similarity is a fundamental challenge in information science. Watanabe's Ugly Duckling Theorem highlights diversity, while algorithmic information theory emphasizes depth through Information Distance. We propose a statistical-mechanical framework that treats program length as energy, with a temperature parameter unifying these two aspects: in the low-temperature limit, similarity approaches Information Distance; in the high-temperature limit, it recovers the indiscriminability of the Ugly Duckling theorem; and at the critical point, it coincides with the Solomonoff prior. We refine the statistical-mechanical framework by introducing regular universal machines and effective degeneracy ratios, allowing us to separate redundant from core diversity. This refinement yields new tools for analyzing similarity and opens perspectives for information distance, model selection, and non-equilibrium extensions.
\end{abstract}
% ----------------------------
% Keywords
% ----------------------------

% ----------------------------
% Introduction
% ----------------------------
\section{Introduction}
Similarity plays a fundamental role in science.
Many advances in physics, information theory, and data science have been achieved by identifying similar structures and constructing unified frameworks that explain them.
Despite its importance, however, defining similarity in a principled way remains a long-standing challenge.

Two contrasting approaches illustrate this difficulty.
Watanabe’s \emph{Ugly Duckling theorem}\cite{watanabe1969,watanabe1985,watanabe1986} shows that if all predicates are treated equally, any pair of objects becomes equally similar, emphasizing the dependence on arbitrary weighting.
In contrast, algorithmic information theory introduces the \emph{Information Distance}\cite{bennett1998,li2008}, defined through Kolmogorov complexity~\cite{kolmogorov1965}, which quantifies similarity by the depth of the shortest description.
These perspectives highlight seemingly opposing aspects---diversity versus depth---and make clear the fundamental difficulty of formalizing similarity.

\medskip
Recent developments in \emph{information thermodynamics}
(e.g., \cite{szilard1929,brillouin1956,jaynes1957,landauer1961,bennett1982,
zurek1989,leff2002,sagawa2012,sagawa2014,parrondo2015,ito2015,ito2018})
and in \emph{algorithmic thermodynamics}
(e.g., \cite{levin1974,chaitin1975,chaitin1987,staiger1998,calude2004partial,
tadaki2002,tadaki2008,baez2012,tadaki2019,ebtekar2025})
provide a natural bridge for importing statistical--mechanical concepts into information theory.
Building on this line of work, we treat program length as energy,
introduce partition functions and free energies,
and define similarity as the thermodynamic work required to couple two objects.
This yields a unified framework where the balance between description length and multiplicity plays the role of energy--entropy tradeoff.

\medskip
Figure~\ref{fig:concept} illustrates this perspective.
The inverse temperature $\beta$ tunes the balance between \emph{depth} (minimal description length) 
and \emph{diversity} (residual entropy of admissible cores).
Because similarity is defined through free--energy differences, 
trivial wrapper redundancies uniformly cancel under the regular-UTM assumption, leaving only core diversity as the entropic contribution.
Within this picture, three classical notions of similarity emerge as temperature regimes:
the Information Distance at low temperature,
the Solomonoff universal prior at $\beta=\ln 2$,
and the Ugly Duckling phase at high temperature.

\begin{figure}[t]
  \centering
  \includegraphics[width=0.75\linewidth]{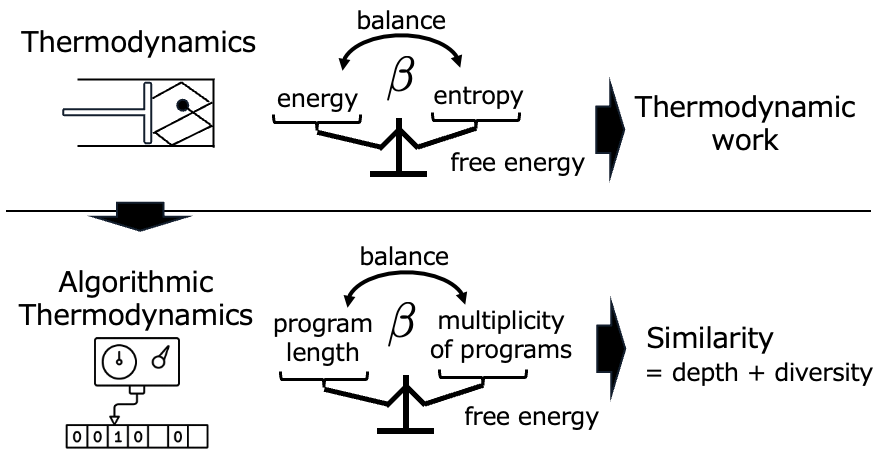}
  \caption{Similarity as a thermodynamic balance. 
Program length plays the role of energy (depth), 
while effective degeneracy provides residual entropy (diversity). 
The tradeoff is governed by the inverse temperature $\beta$, 
interpolating between three regimes: 
Information Distance at low temperature, 
the Solomonoff prior at $\beta=\ln 2$, 
and the Ugly Duckling phase at high temperature.}
  \label{fig:concept}
\end{figure}

\medskip
This paper makes the following contributions:
\begin{itemize}
\item We propose a statistical--mechanical framework that redefines similarity as thermodynamic work.
\item We demonstrate how the Information Distance, the Solomonoff universal prior, and the Ugly Duckling theorem arise in distinct temperature regimes.
\item We introduce the notion of \emph{regular UTMs}, separating trivial wrapper redundancy from genuine core diversity, and define \emph{effective degeneracy} as a new measure of diversity.
\item We identify residual entropy ratios as a structural source of similarity, complementing description length.
\item We outline possible non--equilibrium extensions via Jarzynski--type relations~\cite{jarzynski1997,crooks1999,jarzynski2011}.
\end{itemize}

% ----------------------------
% Preliminaries
% ----------------------------
\section{Preliminaries}

In this section we briefly recall the basic notions from algorithmic information theory
and from Watanabe’s combinatorial argument that are relevant to our framework.
For algorithmic information theory, standard references include Kolmogorov’s
original work~\cite{kolmogorov1965}, Chaitin’s program-size complexity~\cite{chaitin1975,chaitin1987}, 
and the textbook by Li and Vitányi~\cite{li2008}.
For the combinatorial perspective, we refer to Watanabe’s papers on the
Ugly Duckling theorem and the principle of the unique choice~\cite{watanabe1969,watanabe1985,watanabe1986}, 
as well as related philosophical discussions by Goodman~\cite{goodman1972} and
methodological debates in taxonomy~\cite{sneath1973}.

\subsection{Kolmogorov complexity and Information Distance}

For a finite binary string $x$, the \emph{Kolmogorov complexity} $K(x)$ is defined 
as the length of the shortest prefix-free program $p$ that outputs $x$
on a fixed UTM $U$~\cite{kolmogorov1965,solomonoff1964a,solomonoff1964b,chaitin1975,li2008}:
\[
K(x) := \min\{\, |p| : U(p)=x \,\}.
\]
Intuitively, $K(x)$ measures the depth of the most concise description of $x$.  
For two strings $x$ and $y$, the conditional complexity $K(y|x)$ is the length
of the shortest program producing $y$ given $x$ as auxiliary input.

Based on these quantities, Bennett, Gács, Li, Vitányi, and Zurek defined the
\emph{Information Distance}~\cite{bennett1998,li2008}:
\[
E(x,y) := \max\{K(x|y),\, K(y|x)\}.
\]
Up to logarithmic precision, this gives a universal metric for similarity.
In our framework, this “low-temperature” regime will reproduce $K(y|x)$ as reversible work.

\subsection{Solomonoff universal prior}

Solomonoff~\cite{solomonoff1964a,solomonoff1964b} introduced a universal distribution
that assigns to each string $x$ the probability
\[
M(x) := \sum_{p: U(p)=x} 2^{-|p|}.
\]
This Bayesian prior aggregates all programs producing $x$, 
weighting shorter ones more heavily. 
Unlike $K(x)$, which depends only on the shortest program,
$M(x)$ collects the entire ensemble.  
In our thermodynamic analogy, $M(x)$ corresponds to a partition function
evaluated at the critical temperature $\beta = \ln 2$.
This interpretation connects algorithmic probability with free-energy ideas
and underlies universal prediction frameworks such as AIXI~\cite{hutter2005}.

\subsection{Ugly Duckling theorem (Watanabe)\label{sec:UDT}}

Watanabe’s \emph{Ugly Duckling theorem}~\cite{watanabe1969,watanabe1985,watanabe1986}
states that if similarity is defined by counting shared predicates,
then every pair of distinct objects is equally similar.
Formally, in a universe $\Omega$, the number of predicates true of both $x$ and $y$ --- equivalently, the number of subsets $O \subseteq \Omega$ containing both --- is $2^{|\Omega|-2}$, independent of $x$ and $y$.
Thus, when all predicates are weighted equally, all pairs of objects become indiscriminable.
The theorem sharpened Goodman’s earlier philosophical claim---that induction depends
on which predicates are privileged~\cite{goodman1972}---into a combinatorial statement.

To connect this symmetry with algorithmic information theory, 
we reinterpret programs as predicates.
In the standard view, $U(p)=x$ means that program $p$ outputs $x$.
In our reformulation, $U(p)$ is seen as a set $O$ of admissible objects,
with $x\in U(p)$ interpreted as “$x$ satisfies predicate $p$”.
As illustrated in Fig.~\ref{fig:UTM_predicate}, this bridges Watanabe’s set-theoretic
formulation with the program-based language of algorithmic information theory.
It also reveals the limitation of the original theorem:
uniformly counting predicates erases distinctions,
but once predicates are weighted by description length,
shorter programs dominate.
This observation motivates the statistical-mechanical framework we develop,
where predicates are treated as programs and similarity arises from ensembles
weighted by their complexity.

\begin{figure}[t]
  \centering
  \includegraphics[width=0.6\linewidth]{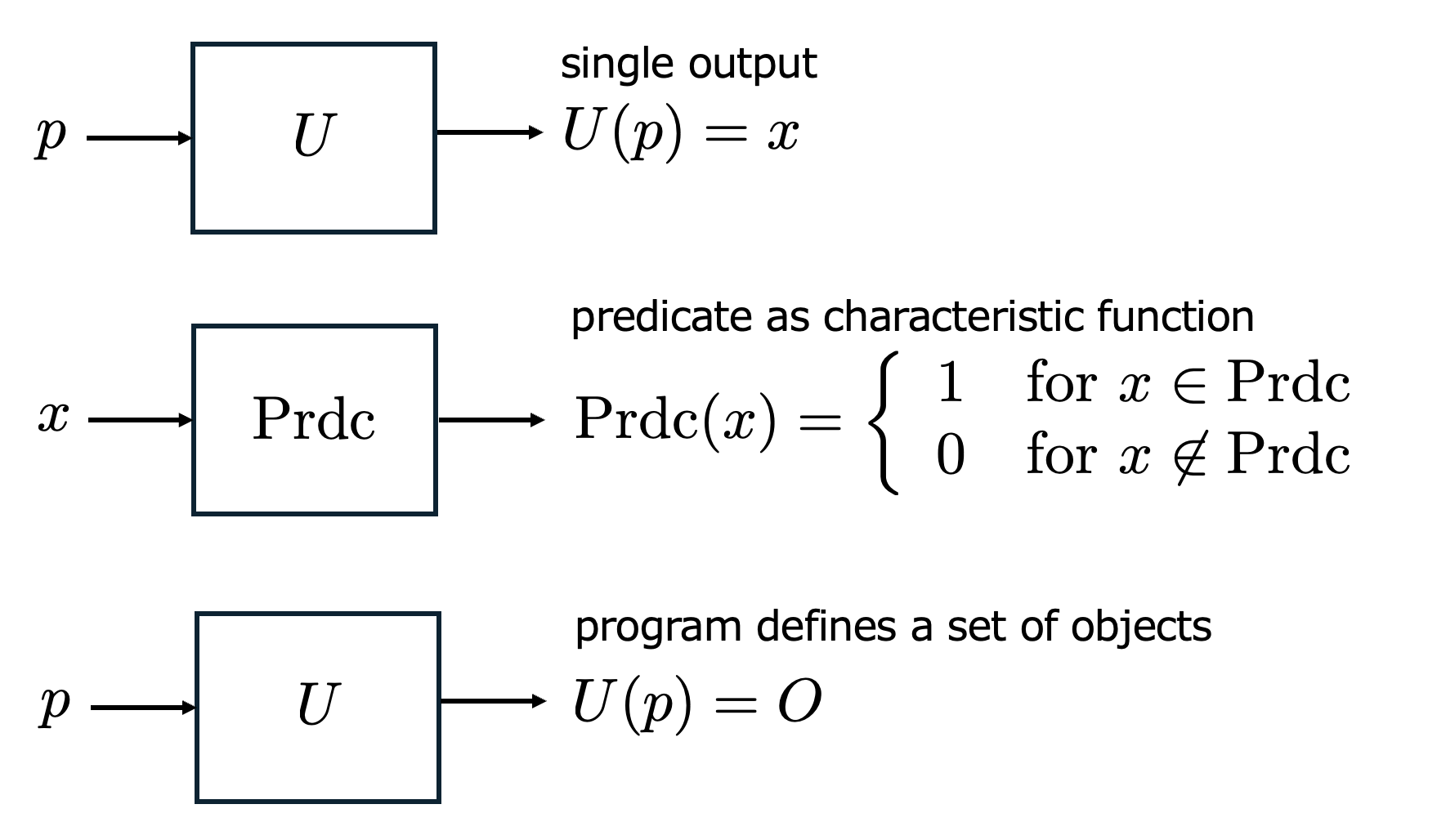}
  \caption{%
Illustration of our reinterpretation.
Top: a program outputs a single object.
Middle: a predicate acts as a characteristic function.
Bottom: in our framework, a program defines a set of objects, enabling predicates to be treated as programs.
}
  \label{fig:UTM_predicate}
\end{figure}
% ----------------------------
% Statistical-mechanical framework
% ----------------------------
\section{Framework}
\subsection{Setup and notation}

As discussed in Sec.~\ref{sec:UDT}, Watanabe’s Ugly Duckling theorem shows that 
if all predicates (i.e.\ subsets of $\Omega$) are counted equally, 
then every pair of distinct objects is equally similar. 
Combinatorially, for a finite universe $\Omega$, 
the number of subsets containing both $x$ and $y$ is always $2^{|\Omega|-2}$, 
independent of the content of $x$ and $y$. 
This symmetry illustrates the arbitrariness of similarity definitions 
when no weighting scheme is imposed.

To overcome this limitation, we introduce a weighting scheme grounded in algorithmic information theory. 
Instead of assigning equal weight to all subsets $O\subseteq\Omega$, 
we weight each subset according to the length of a program that generates it. 
Shorter programs are assigned larger weights, 
so the ensemble is biased toward simpler predicates. 
This construction interpolates between the uniform view of the Ugly Duckling theorem 
and the minimal-description view of Kolmogorov complexity, 
and provides the basis for our statistical--mechanical framework.

\subsubsection{Objects}
The objects under consideration are elements of a finite set $\Omega$, 
where each element is regarded as a finite binary string (data).  
For later use we introduce the set
\begin{equation}
\mathbb{O}^{(-)} := \{\, O \mid O \subseteq \Omega,\; O \neq \emptyset \,\},
\end{equation}
namely the non-empty subsets of $\Omega$.  
Throughout, $x,y$ always denote elements of $\Omega$.  
When necessary we take $O \in \mathbb{O}^{(-)}$ and consider a program satisfying $U(p)=O$.  
For example, in $Z_\beta(x)$ the relevant subsets $O$ are those containing $x$, 
while in $Z_\beta(x,y)$ they are those containing both $x$ and $y$.

\subsubsection{Programs}
As in algorithmic information theory, we consider prefix-free programs $p$ executed on UTM $U$.  
When $U$ outputs a set $O \in \mathbb{O}^{(-)}$ on input $p$, we write $U(p)=O$.  
If $x\in O$, we write
\begin{equation}
U(p) \ni x.
\end{equation}
Similarly, if $x,y \in O$ we write
\begin{equation}
U(p) \ni (x,y).
\end{equation}
The length of a program $p$ is denoted by $|p|$.

\subsection{Thermodynamic setting}

Once predicates are reinterpreted as programs generating subsets $O \subseteq \Omega$, 
we can introduce a statistical--mechanical formalism. 
Partition functions and reversible work are defined in the standard sense of 
statistical mechanics~\cite{pathria2011,huang1987,callen1985}. 
The key idea is to regard the program length $|p|$ as energy, 
and to assign each program a Boltzmann weight $e^{-\beta |p|}$ at inverse temperature $\beta$. 
In this way, short programs dominate at low temperature, 
while at high temperature all programs contribute almost equally, 
recovering the uniform-counting symmetry of the Ugly Duckling theorem.

\subsubsection{Partition functions}
This analogy leads to the following definitions of partition functions and free energies:
\begin{equation}\label{eq:Z(x)}
Z_\beta(x) := \sum_{p: U(p)\ni x} e^{-\beta |p|}, 
\quad 
F_\beta(x) := -\beta^{-1}\ln Z_\beta(x),
\end{equation}
and for two objects $x,y \in \Omega$,
\begin{equation}\label{eq:Z(x,y)}
Z_\beta(x,y) := \sum_{p: U(p)\ni (x,y)} e^{-\beta |p|}, 
\quad 
F_\beta(x,y) := -\beta^{-1}\ln Z_\beta(x,y).
\end{equation}
Here $Z_\beta(x)$ and $Z_\beta(x,y)$ represent partition functions restricted to subsets that contain $x$ or $(x,y)$, respectively.

\subsubsection{Thermodynamic work}
We next consider a process that interpolates between the constraint ``$x$ is included'' and ``both $x$ and $y$ are included.’’  
Introducing a coupling parameter $\lambda \in [0,1]$, we define the energy of each program as
\begin{equation}
E_\lambda(p) = |p| + J \lambda \bigl(1-s_y(p)\bigr),
\end{equation}
where $s_y(p)=1$ if $U(p)\ni y$ and $s_y(p)=0$ otherwise.  
The constant $J$ specifies the strength of the constraint.  
Taking $J\to\infty$ enforces a \emph{hard constraint} at $\lambda=1$, 
so that only programs with $U(p)\ni (x,y)$ contribute.

The corresponding generalized partition function and free energy are
\begin{equation}
\hat{Z}_\beta(\lambda) := \sum_{p:U(p)\ni x} e^{-\beta E_\lambda(p)}, 
\quad 
\hat{F}_\beta(\lambda) := -\beta^{-1} \ln \hat{Z}_\beta(\lambda).
\end{equation}
In the limits we obtain
\begin{equation}
\hat{F}_\beta(0) = F_\beta(x), 
\quad 
\hat{F}_\beta(1) = F_\beta(x,y).
\end{equation}

For this constrained evolution, the generalized force is
\begin{equation}
\Phi_\lambda := \Big\langle \frac{\partial E_\lambda(p)}{\partial \lambda} \Big\rangle_\lambda 
= J\langle 1-s_y(p) \rangle_\lambda,
\end{equation}
and the reversible isothermal work required for the transformation is
\begin{equation}
W^{(\beta)}_{\mathrm{rev}}(x\to x,y) 
= \int_{0}^{1} \Phi_\lambda \, d\lambda
= F_\beta(x,y) - F_\beta(x).
\end{equation}
Thus, similarity is identified with the free energy difference needed to add $y$ to $x$.

\begin{figure}[t]
  \centering
  \includegraphics[width=0.8\linewidth]{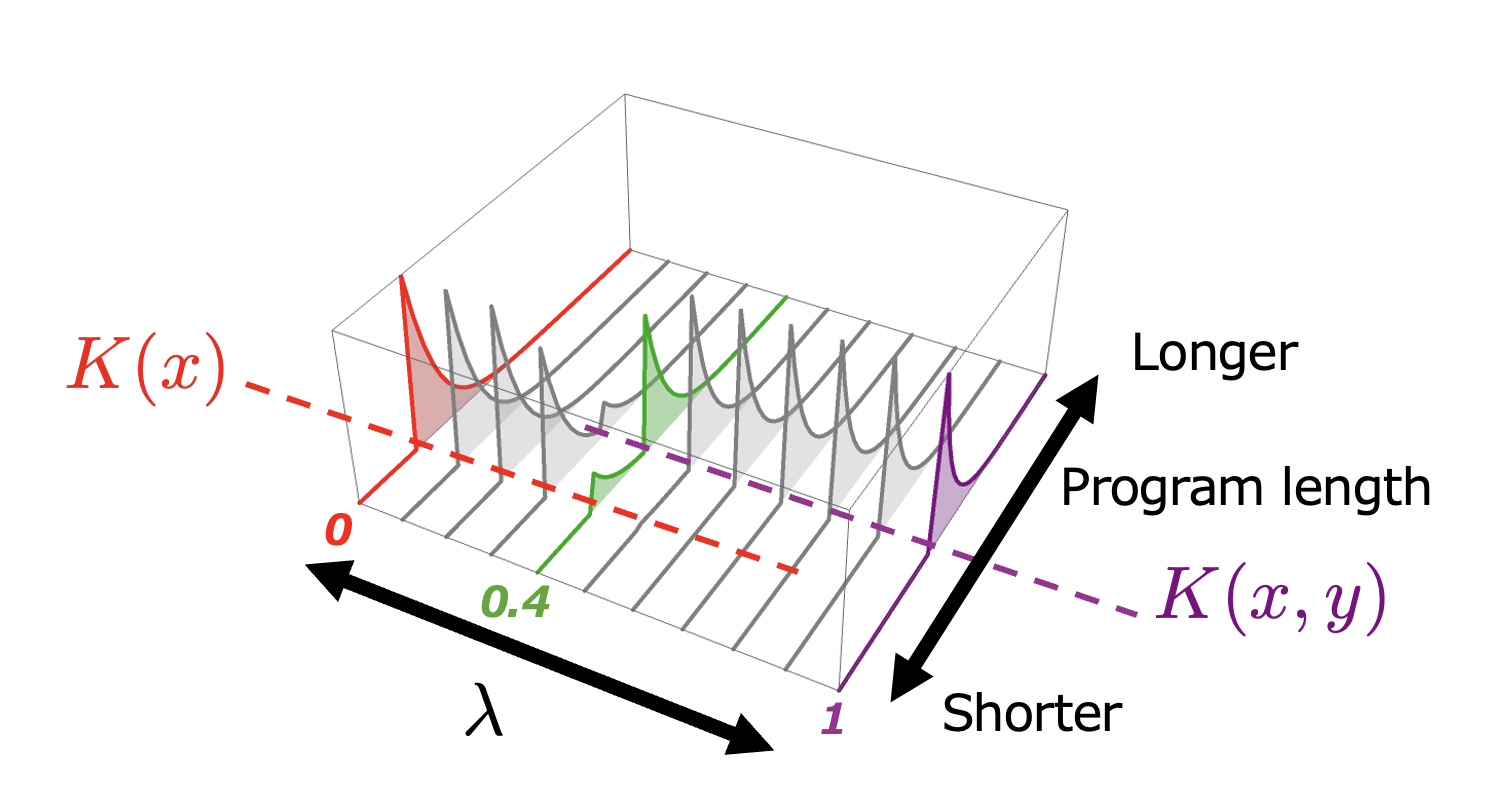}
  \caption{llustration of the $\lambda$-coupling process. The horizontal axis represents program length, the depth axis the coupling parameter $\lambda$, and the vertical peaks indicate weighted program counts. At $\lambda=0$ (red), programs producing $x$ dominate with minimal length $K(x)$. As $\lambda$ increases (e.g. $\lambda=0.4$, green), the ensemble gradually shifts. At $\lambda=1$ (purple), only programs producing both $x$ and $y$ remain, with minimal length $K(x,y)$. The number of shortest programs themselves does not change with $\lambda$; only their normalized weight in the ensemble is rebalanced.}
  \label{fig:lambda_coupling}
\end{figure}

This process is illustrated schematically in Fig.~\ref{fig:lambda_coupling}. 
As $\lambda$ increases, the ensemble shifts smoothly from programs describing $x$ alone 
to programs describing both $x$ and $y$. 
The reversible work quantifies this shift, and we adopt it as our definition of similarity. 
In the following sections we analyze this quantity in different temperature regimes, 
showing how classical notions such as Information Distance and the Ugly Duckling theorem 
emerge as limiting cases.

\section{Three temperature regimes}

\begin{figure}[t]
  \centering
  \includegraphics[width=0.75\linewidth]{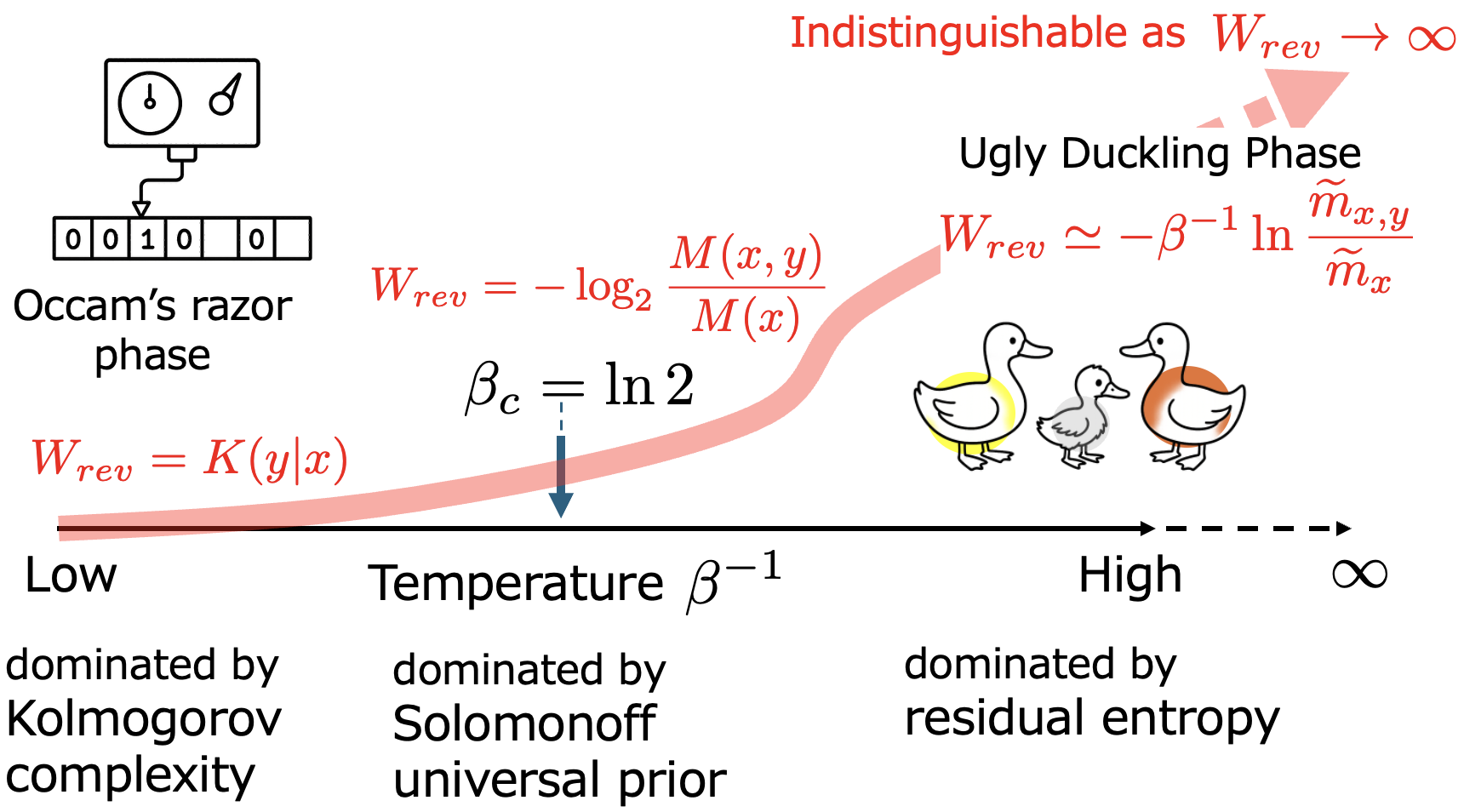}
  \caption{
Schematic phase diagram of similarity as thermodynamic work. 
At low temperature, the reversible work reduces to the conditional complexity $K(y|x)$, 
reproducing the Information Distance (\emph{Occam’s razor phase}). 
At the critical point $\beta=\ln 2$, it coincides with the Solomonoff universal prior. 
At high temperature, the divergent residual--entropy term dominates, 
recovering the essence of Watanabe’s Ugly Duckling theorem. 
Finite conditional complexity contributions $K(y|x)$ remain present but become negligible compared to the divergence.
  }
  \label{fig:temperature_regimes}
\end{figure}

The three regimes are summarized schematically in Fig.~\ref{fig:temperature_regimes}. 
In the following subsections, we analyze them in detail: 
the low--temperature limit corresponding to the Information Distance, 
the critical point $\beta=\ln 2$ yielding the Solomonoff universal prior, 
and the high--temperature limit where residual entropy dominates.

\subsection{Low-temperature limit ($\beta\to \infty$): Occam’s razor phase}
In the limit $\beta \to \infty$, the Boltzmann weight $e^{-\beta |p|}$ suppresses all but the shortest program.  
The partition functions then reduce to
\begin{equation}
Z_\beta(x) \simeq e^{-\beta K(x)}, 
\qquad
Z_\beta(x,y) \simeq e^{-\beta K(x,y)}.
\end{equation}
Consequently, the free energies become
\begin{equation}
F_\beta(x) \simeq K(x), 
\qquad 
F_\beta(x,y) \simeq K(x,y),
\end{equation}
up to $O(\beta^{-1})$ corrections.  
The reversible work is then
\begin{equation}
W^{(\beta\to\infty)}_{\mathrm{rev}}(x\to x,y) 
\simeq K(x,y)-K(x) = K(y|x).
\end{equation}

Thus in the low-temperature limit, reversible work coincides with the conditional Kolmogorov complexity and reproduces the \emph{Information Distance}
\[
D(x,y)=\max\{K(x|y),K(y|x)\}.
\]
This regime reflects the dominance of the shortest description, embodying \emph{Occam’s razor} in algorithmic form, and will be referred to as the ``\emph{Occam’s razor phase}.''

\subsection{Critical point ($\beta=\ln 2$): Bayesian bridge}\label{sec:solomonoff}
At $\beta=\ln 2$, the partition function $Z_{\ln 2}(x)$ coincides with the Solomonoff universal prior $M(x)$.  
This prior is defined as
\begin{equation}
M(x) = \sum_{p:U(p)\ni x} 2^{-|p|},
\end{equation}
a weighted sum over all programs that output $x$, each with weight $2^{-|p|}$.  
Since $U$ is prefix-free, Kraft’s inequality ensures
\[
\sum_{p} 2^{-|p|} \leq 1,
\]
so $M(x)$ is a valid semimeasure and can serve as a universal Bayesian prior over all computable models.

In this case, the reversible work becomes
\begin{equation}
W^{(\ln 2)}_{\mathrm{rev}}(x\to x,y)
= -\log_{2}\frac{M(x,y)}{M(x)}.
\end{equation}
Therefore our similarity measure is interpreted as the free energy cost of adding $y$ given $x$ under the universal prior, which information-theoretically corresponds to the amount of information required for a universal Bayesian update.  
For this reason, we refer to the critical regime as the ``\emph{Bayesian bridge},'' connecting algorithmic and probabilistic perspectives.

\subsection{High-temperature limit: Ugly Duckling phase}

In the opposite limit $\beta\to 0$, all programs contribute almost equally.
Let $N_\ell(x)$ be the number of programs of length $\ell$ with $U(p)\ni x$, and $N_\ell(x,y)$ the number producing $(x,y)$.
Then
\[
Z_\beta(x)\simeq \sum_\ell N_\ell(x),
\qquad
Z_\beta(x,y)\simeq \sum_\ell N_\ell(x,y).
\]
Since $N_\ell(x)\sim 2^\ell$ grows exponentially with $\ell$, these sums diverge for $\beta<\ln 2$, with the critical point $\beta=\ln 2$ marking the convergence boundary.

To regularize the divergence, we introduce an \emph{excess cutoff}: for each $x$, only programs up to $K(x)+\Lambda$ are included.
The corresponding reversible work is
\[
W_{\mathrm{rev}}^{(\Lambda,\beta)}(x\!\to\!x,y)
:=-\frac{1}{\beta}\log\frac{Z_\beta^{(\le K(x,y)+\Lambda)}(x,y)}{Z_\beta^{(\le K(x)+\Lambda)}(x)}.
\]

Under the assumptions of Appendix~\ref{appendix regular UTM}, the high-temperature expansion yields
\[
W_{\mathrm{rev}}^{(\Lambda,\beta)}(x\!\to\!x,y)
= K(y|x) - \beta^{-1}\ln\!\frac{\widetilde m_{x,y}^{(\Lambda)}}{\widetilde m_{x}^{(\Lambda)}} + o(1),
\qquad \beta\to 0,
\]
where the \emph{effective degeneracy}
\[
\widetilde m_o^{(\Lambda)} = m_o + \sum_{\Delta\ge 1}^\Lambda m_o^{(\Delta)}\,2^{-\alpha\Delta},\qquad\mbox{or}\qquad
\widetilde m_o = \lim_{\Lambda\to\infty }\widetilde m_o^{(\Lambda)} 
\]
denotes the effective multiplicity of admissible cores beyond trivial wrappers, with $m_o^{(\Delta)}$ denoting the number of alternative cores of length $K(o)+\Delta$ (see Appendix~\ref{appendix regular UTM} for a discussion of convergence).
Thus $\ln \widetilde m_o$ plays the role of a residual entropy.

The reversible work therefore separates into two contributions:
\begin{itemize}
\item A finite term $K(y|x)$, reflecting the \emph{depth} of description.
\item A divergent term $-{\beta}^{-1}\ln(\widetilde m_{x,y}^{(\Lambda)}/\widetilde m_{x}^{(\Lambda)})$, reflecting the \emph{diversity} of admissible cores.
\end{itemize}

In the strict $\beta\to 0$ limit, the divergent contribution overwhelms all finite terms:
\[
W_{\mathrm{rev}}^{(\Lambda,\beta)}(x\!\to\!x,y)\;\sim\;-{\beta}^{-1}\ln(\widetilde m_{x,y}^{(\Lambda)}/\widetilde m_{x}^{(\Lambda)}), \qquad \beta\to 0,
\]
so every finite difference such as $K(y|x)$ is suppressed.
This universal divergence is the algorithmic analogue of Watanabe’s Ugly Duckling theorem: at infinite temperature, all objects become indiscriminable.
Nevertheless, at small but finite $\beta$ the residual entropy ratios $\widetilde m_{x,y}^{(\Lambda)}/\widetilde m_{x}^{(\Lambda)}$ still differentiate objects, leaving traces of structural information.
In this sense, UDT appears as the $\beta=0$ endpoint, while the nearby high-temperature regime reveals a refined picture of diversity.

Finally, note that if descriptional complexity is ignored—that is, if all implementations are counted equally without weighting by length—our construction reduces exactly to Watanabe’s original setting where all predicates are equally weighted.
Retaining complexity refines this picture: the Ugly Duckling phase is still recovered at $\beta=0$, but finite $\beta$ corrections now encode object-dependent diversity through effective degeneracy ratios.

% ----------------------------
% Discussion and outlook
% ----------------------------
\section{Discussion and Conclusion}

We have presented a statistical--mechanical framework that unifies two seemingly contradictory notions of similarity: 
the depth emphasized by Kolmogorov complexity and the Information Distance, and the diversity emphasized by Watanabe’s Ugly Duckling theorem.  
By interpreting program length as energy and similarity as thermodynamic work, we showed that these two aspects are connected through a single temperature parameter.  
In the low--temperature limit, similarity reduces to conditional Kolmogorov complexity, while at high temperature it is dominated by effective degeneracy, reproducing the essence of the Ugly Duckling theorem.  
At the critical point $\beta=\ln 2$, the Solomonoff universal prior emerges as the universal Bayesian prior.
In particular, our analysis clarifies the relation to Watanabe’s theorem.
If descriptional complexity is ignored altogether—that is, if every admissible implementation is counted equally without weighting by program length—our construction reduces exactly to Watanabe’s original setting where all predicates are equally weighted.
In contrast, keeping length weighting but taking the thermodynamic limit $\beta \to 0$ produces a universal divergence that likewise erases all finite differences, again recovering the indiscriminability of the Ugly Duckling theorem.
The two perspectives are not identical: the first modifies the weighting scheme, while the second arises as a limiting case.
Yet both capture the same essential symmetry, and our framework embeds them within a single parameterized picture.
Retaining complexity, however, refines this picture: the strict $\beta=0$ endpoint reproduces indiscriminability, while at small finite $\beta$ residual entropy ratios $\ln(\widetilde m_{x,y}/\widetilde m_{x})$ survive and provide additional structural information.
Thus our construction not only recovers Watanabe’s result in the complexity-free case and at $\beta=0$, but also extends both perspectives into a more nuanced setting where depth and diversity jointly contribute.

A central novelty of this work lies in clarifying that *not all program multiplicities are equal*.  
Wrappers, which correspond to trivial syntactic redundancy, proliferate exponentially but uniformly across all objects.  
As such, they contribute no distinguishing information.  
In contrast, the multiplicity of distinct *cores*---fundamentally different computational routes to the same object---remains object--dependent.  
Crucially, distinct cores are not compressible into a single shortest form: each represents an essentially different constructive pathway.  
By separating wrapper redundancy from core diversity through the notion of \emph{regular UTMs}, we defined the \emph{effective degeneracy} $\widetilde m_o$, which aggregates only the meaningful multiplicities.  
Its logarithm plays the role of a residual entropy, complementing description length.  
This distinction shows that similarity naturally decomposes into two contributions:  
the depth of minimal description, $K(y|x)$, and the diversity of admissible cores, $\ln(\widetilde m_{x,y}/\widetilde m_{x})$.

\begin{figure}[t]
  \centering
  \includegraphics[width=0.7\linewidth]{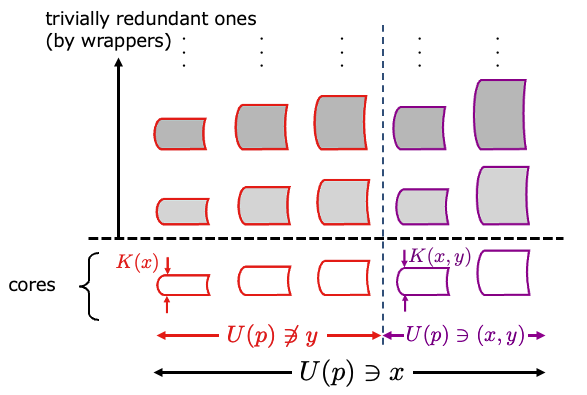}
\caption{%
Instruction--stack analogy for programs.  
The lowest block represents the shortest possible instruction implementing $x$ or $(x,y)$, 
with thickness equal to $K(x)$ or $K(x,y)$.  
Additional layers correspond to full instruction variants obtained by adding redundant wrappers 
on top of the same core.  
Such wrappers proliferate uniformly across all objects and therefore cancel out in the thermodynamic 
formulation, leaving only the multiplicity of distinct cores as the residual entropic contribution.%
}
  \label{fig:instruction_stack}
\end{figure}

The analogy illustrated in Fig.~\ref{fig:instruction_stack}  highlights that similarity goes beyond what is captured by minimal description length alone.
Traditional complexity--based approaches compare only the thickness of the thinnest sheet. 
Our framework instead emphasizes that tasks are not determined solely by one minimal instruction, but also by the variety of distinct instructions that can realize them. 
This residual diversity---captured by effective degeneracy---aligns more closely with the everyday intuition 
that two tasks can feel similar not only when their minimal descriptions are close, 
but also when they admit a wide range of interchangeable realizations.

To our knowledge, this is the first framework that systematically isolates \emph{core diversity} (as opposed to trivial wrapper redundancy) and formalizes its contribution to similarity.  
Algorithmic information theory has traditionally focused almost exclusively on description length and shortest programs, with little attention to multiplicities.  
The present refinement therefore extends AIT by highlighting core diversity as a genuine and irreducible contributor to similarity.  
In particular, the high--temperature regime is governed not by minimal descriptions but by residual entropy differences, with effective degeneracy ratios determining whether objects appear distinguishable or indiscriminable.  

It should be stressed, however, that effective degeneracy $\widetilde m_o$ cannot be computed exactly.  
This limitation follows directly from the uncomputability of Kolmogorov complexity: if $K(x)$ itself is not effectively calculable, then enumerating all cores and their multiplicities is infeasible.  
Moreover, degeneracy depends on the chosen UTM, and the number of programs grows exponentially with length, further constraining explicit evaluation.  
Nevertheless, just as residual entropy in statistical mechanics captures structural properties of ground states that cannot be enumerated microscopically~\cite{pathria2011}, effective degeneracy provides a principled way to represent the contribution of diversity to similarity.  
Approximate estimates may be obtained in practice through surrogate methods such as compression--based approximations, minimum description length principles, or bounded searches within finite cutoffs.  
In this sense, the limitations of effective degeneracy mirror those of Kolmogorov complexity itself: both are uncomputable in the strict sense and machine--dependent, yet both provide robust and widely accepted theoretical foundations.

In addition, effective degeneracy is genuinely dependent on the choice of regular UTM. Unlike Kolmogorov complexity, where UTM dependence is confined to an additive constant, $\widetilde m_o$ can vary structurally across machines, since different UTMs expose different families of admissible cores. This dependence should not be viewed as a drawback: rather, it reflects the natural fact that similarity is defined relative to a measuring device—in this case, the universal machine. Once trivial wrappers are factored out, the contribution of core diversity is stable within any fixed regular UTM, providing a consistent and interpretable measure of diversity.

Beyond its theoretical contribution, our framework suggests possible implications for model selection and learning theory, where weighting schemes and residual entropy differences naturally arise.  
The degeneracy contribution can be interpreted as a measure of ensemble diversity, complementing the description--length term that implements Occam’s razor.  
This perspective offers a principled way to balance depth and diversity in practical learning scenarios.  

Empirical similarity measures also fall into place within this perspective.  
For example, the \emph{Normalized Google Distance} (NGD)~\cite{cilibrasi2007} and the 
\emph{Normalized Compression Distance} (NCD)~\cite{li2004}  
have traditionally been regarded as practical approximations to the Information Distance.  
Within our framework, however, they acquire a natural reinterpretation.  
NGD relies on frequency counts $f(x)$ and $f(x,y)$ obtained from web search engines.  
These counts effectively represent the number of distinct contexts in which $x$ and $y$ appear, and can therefore be viewed as a practical proxy for effective degeneracy.  
In thermodynamic terms, NGD corresponds to a diversity--dominated, finite-- or high--temperature regime.  
By contrast, NCD, defined via general--purpose compressors, is more naturally viewed as a low--temperature proxy: shortest descriptions dominate, and compression length serves as a practical surrogate for Kolmogorov complexity.  
Thus NGD and NCD, while both approximations to the Information Distance, can be reinterpreted as complementary endpoints of the same thermodynamic axis:  
NGD reflecting diversity at high temperature and NCD reflecting depth at low temperature.  
The practical acceptance of NGD as a measure of similarity suggests that capturing similarity through effective degeneracy ratios is indeed meaningful, providing empirical support for the theoretical framework developed here.

\medskip
Beyond these technical aspects, the thermodynamic reformulation itself has several conceptual advantages: 
\begin{itemize}
  \item It places two apparently contradictory perspectives---depth and diversity---on a single temperature axis, making their relation transparent.  
  \item It expresses similarity as a free--energy difference, a universal physical quantity, thereby providing a common ground for otherwise disparate notions.  
  \item It enables the use of tools from statistical mechanics, such as Jarzynski--type relations and phase--transition analogies, to study similarity.  
  \item It offers an intuitive picture, where low temperature corresponds to Occam’s razor (minimal descriptions) and high temperature to the indiscriminability of the Ugly Duckling theorem.  
   \item The statistical-mechanical formulation also suggests natural extensions beyond the classical setting, potentially including quantum ensembles.
\end{itemize}
This combination of mathematical rigor and physical intuition may help broaden the conceptual reach of algorithmic information theory and its applications.

Finally, our approach is conceptually related to the line of work known as ``algorithmic thermodynamics,'' 
where program length is treated as energy and partition functions over programs are analyzed.  
That literature has primarily focused on the thermodynamic interpretation of Kolmogorov complexity and universal distributions, 
such as using free energy to describe compression limits.  
Our contribution differs in that we explicitly apply the thermodynamic formalism to the problem of similarity, 
showing how two seemingly opposed notions---the Information Distance (depth) and the Ugly Duckling theorem (diversity)---emerge as limiting phases within a single statistical--mechanical framework.  
In particular, the identification of \emph{effective degeneracy}, purified of wrapper redundancy, 
as the dominant driver of similarity in the high--temperature limit appears to be new.  
This highlights how the algorithmic--thermodynamic perspective can reveal structural insights 
that are not visible in either combinatorial or information--theoretic treatments alone.  

Metaphorically, the tension between depth and diversity recalls a contrast in art: classical traditions aim for faithful compression of form, while impressionism values the richness of many perspectives. This analogy, albeit informal, conveys the intuitive balance between compactness and multiplicity that our framework unifies.
 % ----------------------------
% Acknowledgments (optional)
% ----------------------------
\section*{Acknowledgments}
% Acknowledge funding, colleagues, etc.
I am grateful to my colleagues at AIST for inspiring the initial motivation for this work and for their various forms of support.
I also acknowledge the use of OpenAI’s ChatGPT-5 in checking logical consistency and improving the English presentation of the manuscript.

% ----------------------------
% Appendices (optional)
% ----------------------------
\appendix
\section*{Appendix: Rigorous results and non-equilibrium extensions}
This appendix provides rigorous justification and computational extensions of the results presented in the main text. 
Appendix~\ref{appendix regular UTM} gives a high--temperature expansion under a regular UTM and the 
\emph{excess cutoff} scheme (also referred to as a surplus cutoff), 
leading to a precise formulation of the reversible work. 
Appendix~\ref{appendix jarzynski} outlines non-equilibrium extensions based on the Jarzynski equality 
(see e.g.~\cite{jarzynski1997,crooks1999,jarzynski2011}), including both theoretical formulation 
and a sketch of numerical implementation.

\section{High--temperature limit under regular UTM and excess cutoff\label{appendix regular UTM}}

\subsection{Preliminaries}

For an object $o$, define
\[
N_\ell(o):=\bigl|\{\,p:\,|p|=\ell,\ U(p)\ni o\,\}\bigr|,
\]
the number of programs of length $\ell$ generating $o$.  
The \emph{ground length} $K(o)$ is the minimal $\ell$ with $N_\ell(o)>0$.

\subsection{Regular UTMs and wrapper construction}

\begin{definition}[Regular UTM]
A \emph{regular Universal Turing Machine (UTM)} is a prefix-free machine $U$
such that every program $p$ admits a unique decomposition $p=w\Vert q$ into
\begin{itemize}
\item a \emph{wrapper} $w$ drawn from a syntactically defined family $\mathcal W$, and
\item a \emph{core code} $q$ executed by a fixed reference machine $U_0$, i.e.\ $U(w\Vert q)=U_0(q)$.
\end{itemize}
The wrapper set acts uniformly and independently of the output.
Let $a_d:=|\{w\in\mathcal W:\ |w|=d\}|$ denote the number of wrappers of length $d$.
\end{definition}

\paragraph{Explicit construction.}
Fix a marker $M$ of length $h\ge 1$.  
Define
\[
\mathcal W_d := \{\, w=sM:\ |w|=d,\ \text{$s$ contains no occurrence of $M$}\,\}.
\]
That is, a wrapper is any binary string that ends with $M$ and whose prefix $s$ is $M$-free.
This guarantees a unique boundary between $w$ and the core $q$: in any concatenation
$p=w\Vert q$, the \emph{first (leftmost) occurrence} of $M$ in $p$ is precisely the end of $w$.
After that boundary, the core $q$ may contain $M$ arbitrarily; parsing remains unambiguous
because no occurrence of $M$ appears \emph{before} the boundary inside $w$.

The set $\{\,s:\ |s|=d-h,\ s\ \text{$M$-free}\,\}$ is a regular language recognized by a finite automaton,
hence its cardinality grows exponentially. Therefore there exist constants $c,\alpha>0$ such that
\[
a_d := |\mathcal W_d| \asymp c\,2^{\alpha d}\qquad(d\to\infty).
\]
We also include the empty wrapper $w=\epsilon$ so that shortest programs ($d=0$) are available.

\subsection{Effective degeneracy}

Let $m_o$ denote the number of shortest cores of length $K(o)$ on $U_0$.  
In addition, there may exist \emph{alternative cores} of length $K(o)+\Delta$ for integers $\Delta\ge 1$.  
Let $m_o^{(\Delta)}$ be their multiplicity.  
Such alternatives include, in particular, the case of zero-length wrappers (pure cores of length $K(o)+\Delta$).  
At total length $K(o)+d$ the decomposition gives
\begin{equation}\label{eq:split}
N_{K(o)+d}(o)
= m_o\,a_d + \sum_{\Delta=1}^{d-1} m_o^{(\Delta)}\,a_{d-\Delta} + E_d(o),
\end{equation}
where $E_d(o)$ collects boundary effects.
These include, for example:
\begin{itemize}
\item cases where an alternative core of length $K(o)+\Delta$ cannot contribute because $d<\Delta$ (negative wrapper length),
\item small-$d$ deviations where the ratio $a_{d-\Delta}/a_d$ has not yet converged to its exponential asymptotic $2^{-\alpha\Delta}$,
\item or finitely many exceptional assignments in the prefix-free code of $U_0$ that break the uniform pattern at very short lengths.
\end{itemize}
All such contributions occur only at finitely many $d$ or remain uniformly bounded,
hence they satisfy $E_d(o)=o(a_d)$ as $d\to\infty$.

For wrapper families with exponential growth $a_d\asymp c\,2^{\alpha d}$ ($0<\alpha<1$), the ratio
\[
\frac{a_{d-\Delta}}{a_d}\;\to\;2^{-\alpha\Delta}\quad(d\to\infty).
\]
Hence the contribution of an alternative core is asymptotically a fixed constant fraction $2^{-\alpha\Delta}$ of the main family.

\begin{definition}[Effective degeneracy]
The \emph{effective degeneracy} of object $o$ is
\[
\widetilde m_o^{(d)} := m_o + \sum_{\Delta\ge 1}^{d-1} m_o^{(\Delta)}\,2^{-\alpha\Delta}.
\]
We define
\[
\widetilde m_o := \lim_{d\to\infty}\widetilde m_o^{(d)},
\]
which is finite whenever the alternative families are at most subexponential
(as ensured by prefix-freeness of $U_0$ restricted to output $o$).

Although $\widetilde m_o$ is formally defined by summing over all alternative cores,
in the high--temperature analysis with an excess cutoff only finitely many near--shortest
alternatives contribute appreciably. Longer alternatives are both exponentially suppressed
(by the factor $2^{-\alpha\Delta}$) and excluded by the cutoff window, so $\widetilde m_o$
is effectively determined by a finite neighborhood around the ground length.
Thus $\ln \widetilde m_o$ may be regarded as a \emph{residual entropy} associated with $o$, 
paralleling ground-state degeneracy in statistical mechanics.
\end{definition}

\subsection{Asymptotics of Program Counts}

\begin{proposition}\label{prop:residual}
For a regular UTM with exponential wrapper growth $a_d\asymp c\,2^{\alpha d}$,
the program count admits the representation
\[
N_{K(o)+d}(o) = \widetilde m_o\,a_d\,(1+r_o(d)),\qquad r_o(d)\to 0\ (d\to\infty).
\]
\end{proposition}

\begin{proof}[Sketch]
Divide \eqref{eq:split} by $a_d$:
\[
\frac{N_{K(o)+d}(o)}{a_d}
= m_o + \sum_{\Delta\ge 1} m_o^{(\Delta)}\,\frac{a_{d-\Delta}}{a_d} + \frac{E_d(o)}{a_d}.
\]
For each fixed $\Delta$, $\tfrac{a_{d-\Delta}}{a_d}\to 2^{-\alpha\Delta}$.  
Thus the finite sum of nearby alternatives converges to
$\sum_{\Delta} m_o^{(\Delta)}2^{-\alpha\Delta}$.  
If alternatives exist for infinitely many $\Delta$, sparsity (subexponential growth)
ensures the tail $\sum_{\Delta>D} m_o^{(\Delta)}a_{d-\Delta}/a_d$ vanishes uniformly as $D\to\infty$.  
Hence
\[
\frac{N_{K(o)+d}(o)}{a_d}\;=\;\widetilde m_o+o(1),
\]
which is equivalent to $N_{K(o)+d}(o)=\widetilde m_o\,a_d(1+r_o(d))$ with $r_o(d)\to 0$.
\end{proof}

\subsection{Partition Functions with Excess Cutoff}

For an excess window $\Lambda\ge 0$, define
\begin{align}
Z_\beta^{(\le K(x)+\Lambda)}(x)
&=e^{-\beta K(x)}\sum_{d=0}^{\Lambda}N_{K(x)+d}(x)\,e^{-\beta d},\\
Z_\beta^{(\le K(x,y)+\Lambda)}(x,y)
&=e^{-\beta K(x,y)}\sum_{d=0}^{\Lambda}N_{K(x,y)+d}(x,y)\,e^{-\beta d}.
\end{align}
By Proposition~\ref{prop:residual},
\begin{align}
Z_\beta^{(\le K(x)+\Lambda)}(x)
&=e^{-\beta K(x)}\,\widetilde m_x^{(\Lambda)} \sum_{d=0}^{\Lambda} a_d(1+r_x(d))e^{-\beta d},\\
Z_\beta^{(\le K(x,y)+\Lambda)}(x,y)
&=e^{-\beta K(x,y)}\,\widetilde m_{x,y}^{(\Lambda)} \sum_{d=0}^{\Lambda} a_d(1+r_{x,y}(d))e^{-\beta d}.
\end{align}
Hence the ratio is
\begin{equation}\label{eq:ratio-core}
\frac{Z_\beta^{(\le K(x,y)+\Lambda)}(x,y)}{Z_\beta^{(\le K(x)+\Lambda)}(x)}
= e^{-\beta(K(x,y)-K(x))}~ \frac{\widetilde m_{x,y}^{(\Lambda)}}{\widetilde m_x^{(\Lambda)}}
~ \frac{\sum_{d=0}^{\Lambda}a_d(1+r_{x,y}(d))e^{-\beta d}}
{\sum_{d=0}^{\Lambda}a_d(1+r_x(d))e^{-\beta d}}.
\end{equation}

\subsection{High-Temperature Expansion}

For fixed $\Lambda$ and $\beta\to 0$,
\[
\sum_{d=0}^{\Lambda}a_d(1+r_o(d))e^{-\beta d}
=\Bigl(\sum_{d=0}^{\Lambda}a_d\Bigr)\,[1+O(\varepsilon(\Lambda))+O(\beta\Lambda)],
\]
where $\varepsilon(\Lambda):=\max\{|\varepsilon_x(\Lambda)|,|\varepsilon_{x,y}(\Lambda)|\}$ and
$\varepsilon(\Lambda)\to 0$ by Proposition~\ref{prop:residual}.  
Thus \eqref{eq:ratio-core} becomes
\[
\frac{Z_\beta^{(\le K(x,y)+\Lambda)}(x,y)}{Z_\beta^{(\le K(x)+\Lambda)}(x)}
= e^{-\beta(K(x,y)-K(x))}\cdot \frac{\widetilde m_{x,y}^{(\Lambda)}}{\widetilde m_x^{(\Lambda)}}\,[1+O(\varepsilon(\Lambda))+O(\beta\Lambda)].
\]

Taking $\Lambda\to\infty$ under the scaling $\beta\Lambda\to 0$, $\varepsilon(\Lambda)\to 0$, the error vanishes.

\subsection{Reversible Work at High Temperature}

Finally,
\begin{align}
W_{\mathrm{rev}}^{(\Lambda,\beta)}(x\!\to\!x,y)
&=-{\beta}^{-1}\log\frac{Z_\beta^{(\le K(x,y)+\Lambda)}(x,y)}{Z_\beta^{(\le K(x)+\Lambda)}(x)}\\
&=[K(x,y)-K(x)] - {\beta}^{-1}\ln\frac{\widetilde m_{x,y}^{(\Lambda)}}{\widetilde m_{x}^{(\Lambda)}} + o(1).
\end{align}
Thus
\[
W_{\mathrm{rev}}^{(\Lambda,\beta)}(x\!\to\!x,y)
=K(y|x) - {\beta}^{-1}\ln\frac{\widetilde m_{x,y}^{(\Lambda)}}{\widetilde m_{x}^{(\Lambda)}}+o(1),
\]
with the interpretation that $K(y|x)$ is the finite \emph{depth} contribution, while
$\ln(\widetilde m_{x,y}^{(\Lambda)}/\widetilde m_{x}^{(\Lambda)})$ is the divergent \emph{diversity} or residual entropy contribution.

\subsection{Cutoff Schemes: Absolute vs Excess}
An alternative cutoff restricts the absolute program length, e.g.\ $|p|\le L$.
While natural and simple, this treats different objects asymmetrically:
the relative position of $L$ to the ground length $K(x)$ may vary greatly,
biasing comparisons across objects.
By contrast, the excess cutoff normalizes each object by its own $K(x)$
and examines only the surplus window $\Lambda$.
This places all objects on equal footing and allows the degeneracy ratio
to emerge cleanly as the dominant contribution in the high--temperature limit.

\subsection{Supplementary Note: Growth of wrapper families}

For completeness we justify the claim that the wrapper family defined by a fixed marker $M$ 
has exponential growth.

\begin{proposition}
Fix a marker $M$ of length $h\ge 1$.  
Let $\mathcal A_n(M)$ denote the set of binary strings of length $n$ that do not contain $M$ as a substring.  
Then $|\mathcal A_n(M)|=C\,\mu^n+o(\mu^n)$ as $n\to\infty$, 
where $\mu\in[1,2)$ is the spectral radius of the transition matrix of the $M$-avoidance automaton, 
and $C>0$ depends on $M$.
\end{proposition}

\begin{proof}[Sketch]
The set $\mathcal A_n(M)$ is a regular language recognized by a finite automaton
constructed from the prefix function of $M$ (e.g.\ the Aho-Corasick automaton).  
Let $A$ be the adjacency matrix of the automaton restricted to safe states.
Then
\[
|\mathcal A_n(M)| = \boldsymbol e_0^\top A^n \boldsymbol 1
\]
where $\boldsymbol e_0$ is the unit vector corresponding to the initial automaton state, and $\boldsymbol 1$ is the all-ones vector selecting all accepting states.
In other words, $\boldsymbol e_0^\top A^n \boldsymbol 1$ counts the number of $n$-step paths of the automaton that start in the initial state and end in any safe state, which is exactly the number of length-$n$ binary strings that avoid $M$.
By Perron-Frobenius, $A$ has a dominant eigenvalue $\mu=\rho(A)$ with positive eigenvectors,
so $|\mathcal A_n(M)|=C\,\mu^n+o(\mu^n)$ with $C>0$.
\end{proof}

As a consequence, the wrapper counts satisfy
\[
a_d = |\mathcal W_d| = |\mathcal A_{d-h}(M)| 
= c\,2^{\alpha d}\,[1+o(1)], \qquad \alpha=\log_2 \mu \in [0,1).
\]
Thus there exist constants $c,\alpha>0$ such that $a_d\asymp c\,2^{\alpha d}$, as used in the main text.

\section{Non-equilibrium extension via Jarzynski equality\label{appendix jarzynski}}

\subsection{Theoretical sketch}
Consider a system in contact with a thermal bath at inverse temperature $\beta$, with distribution
\[
\pi_\beta(p \mid \lambda)
=\frac{e^{-\beta E_\lambda(p)}}{\hat{Z}_\beta(\lambda)},
\]
where $E_\lambda(p)$ depends on a control parameter $\lambda\in[0,1]$.  
In the reversible case, slow variation of $\lambda$ gives work equal to free energy difference.  
In the irreversible case, extra work appears, but Jarzynski’s equality ensures
\[
\bigl\langle e^{-\beta W}\bigr\rangle
= e^{-\beta \Delta F}.
\]

With the excess cutoff, admissible programs are restricted at each $\lambda$, leading to the partition functions
\[
\hat{Z}^{(\Lambda)}_\beta(\lambda)
=\sum_{p\in\mathcal{P}_\Lambda(\lambda)} e^{-\beta E_\lambda(p)}.
\]
where $\mathcal{P}_\Lambda(\lambda)$ denotes the set of admissible programs under the excess cutoff $\Lambda$ at control parameter $\lambda$.
The Jarzynski equality then reads
\[
\bigl\langle e^{-\beta W}\bigr\rangle
=\frac{\hat{Z}^{(\Lambda)}_\beta(1)}{\hat{Z}^{(\Lambda)}_\beta(0)}
=\exp\!\bigl[-\beta\,\Delta F^{(\Lambda)}_\beta\bigr],
\]
which recovers the reversible work defined in the main text.

\subsection{Numerical protocol (sketch)}

A numerical protocol for estimating similarity at finite temperature is as follows:
\begin{enumerate}
\item Enumerate all programs up to surplus cutoff $\Lambda$ that output $x$; identify those that also output $y$.
\item Discretize $\lambda$ into a sequence $0=\lambda_0<\cdots<\lambda_M=1$.
\item At each $\lambda_k$, sample programs from $\pi_\beta(p\mid \lambda_k)\propto e^{-\beta E_{\lambda_k}(p)}$.
\item For each sampled trajectory, accumulate the work $W$.
\item Estimate $\exp(-\beta\Delta F)$ via the average $\langle e^{-\beta W}\rangle$.
\end{enumerate}

This protocol provides a practical scheme for estimating the finite-temperature similarity measure defined in this work.

\paragraph{Remarks on feasibility and interpretation.}
It should be stressed that the above protocol does not require the exact values of 
$K(x)$ or $K(x,y)$, which are uncomputable in general. 
In practice, one works with finite surrogate sets of programs, 
obtained for example from compression-based approximations, 
minimum description length principles, or bounded searches up to a cutoff length. 
Within such restricted ensembles, the Jarzynski procedure remains valid and yields 
an effective estimate of the free energy difference. 
Thus the non-equilibrium extension should be understood not as a way to compute 
the exact Kolmogorov-based similarity, but as a principled framework for defining 
and estimating approximate similarity measures.

An additional conceptual benefit is that Jarzynski’s framework makes explicit the gap 
between reversible and irreversible processes. 
While our similarity measure is formally defined by the free energy difference between 
$\lambda=0$ and $\lambda=1$, practical approximations, coarse sampling, or limited 
exploration of program space will generally produce an average work 
$\langle W \rangle \ge \Delta F$.  
The excess $\langle W \rangle - \Delta F$ can be interpreted as an 
\emph{irreversible cost of similarity}, quantifying how much apparent dissimilarity 
is introduced by computational or sampling limitations.  
This provides a natural measure of the reliability of approximate similarity estimates.

% ----------------------------
% References
% ----------------------------
\bibliographystyle{unsrtnat} % or unsrtnat, abbrvnat
\bibliography{refs}

\begin{thebibliography}{41}
\providecommand{\natexlab}[1]{#1}
\providecommand{\url}[1]{\texttt{#1}}
\expandafter\ifx\csname urlstyle\endcsname\relax
  \providecommand{\doi}[1]{doi: #1}\else
  \providecommand{\doi}{doi: \begingroup \urlstyle{rm}\Url}\fi

\bibitem[Watanabe(1969)]{watanabe1969}
S.~Watanabe.
\newblock \emph{Knowing and Guessing: A Quantitative Study of Inference and
  Information}.
\newblock Quantitative Study of Inference and Information. Wiley, 1969.
\newblock ISBN 9780471921301.
\newblock URL \url{https://books.google.co.jp/books?id=gIxQAAAAMAAJ}.

\bibitem[Watanabe(1985)]{watanabe1985}
S.~Watanabe.
\newblock \emph{Pattern Recognition: Human and Mechanical}.
\newblock Wiley, 1985.
\newblock ISBN 9780471808152.
\newblock URL \url{https://books.google.co.jp/books?id=ZmdQAAAAMAAJ}.

\bibitem[Watanabe(1986)]{watanabe1986}
Satosi Watanabe.
\newblock Epistemological relativity.
\newblock \emph{Annals of the Japan Association for Philosophy of Science},
  7\penalty0 (1):\penalty0 1--14, 1986.
\newblock \doi{10.4288/jafpos1956.7.1}.

\bibitem[Bennett et~al.(1998)Bennett, Gacs, Li, Vitanyi, and
  Zurek]{bennett1998}
C.H. Bennett, P.~Gacs, Ming Li, P.M.B. Vitanyi, and W.H. Zurek.
\newblock Information distance.
\newblock \emph{IEEE Transactions on Information Theory}, 44\penalty0
  (4):\penalty0 1407--1423, 1998.
\newblock \doi{10.1109/18.681318}.

\bibitem[Li and Vit{\'a}nyi(2008)]{li2008}
Ming Li and Paul Vit{\'a}nyi.
\newblock {An Introduction to Kolmogorov Complexity and Its Applications}.
\newblock \emph{Springer}, 2008.

\bibitem[Kolmogorov(1965)]{kolmogorov1965}
A.~N. Kolmogorov.
\newblock Three approaches to the definition of the concept ``quantity of
  information''.
\newblock \emph{Probl. Peredachi Inf}, 1\penalty0 (1):\penalty0 3--11, 1965.

\bibitem[Szilard(1929)]{szilard1929}
Leo Szilard.
\newblock On the decrease of entropy in a thermodynamic system by the
  intervention of intelligent beings.
\newblock \emph{Zeitschrift f{\"u}r Physik}, 53:\penalty0 840--856, 1929.
\newblock \doi{10.1007/BF01341281}.

\bibitem[Brillouin(1956)]{brillouin1956}
Leon Brillouin.
\newblock \emph{Science and Information Theory}.
\newblock Academic Press, New York, 1956.

\bibitem[Jaynes(1957)]{jaynes1957}
E.~T. Jaynes.
\newblock Information theory and statistical mechanics.
\newblock \emph{Physical Review}, 106\penalty0 (4):\penalty0 620--630, 1957.
\newblock \doi{10.1103/PhysRev.106.620}.

\bibitem[Landauer(1961)]{landauer1961}
R.~Landauer.
\newblock Irreversibility and heat generation in the computing process.
\newblock \emph{IBM Journal of Research and Development}, 5\penalty0
  (3):\penalty0 183--191, 1961.
\newblock \doi{10.1147/rd.53.0183}.

\bibitem[Bennett(1982)]{bennett1982}
Charles~H. Bennett.
\newblock {The Thermodynamics of Computation---A Review}.
\newblock \emph{International Journal of Theoretical Physics}, 21\penalty0
  (12):\penalty0 905--940, 1982.
\newblock \doi{10.1007/BF02084158}.

\bibitem[Zurek(1989)]{zurek1989}
W.~H. Zurek.
\newblock Thermodynamic cost of computation, algorithmic complexity and the
  information metric.
\newblock \emph{Nature}, 341\penalty0 (6238):\penalty0 119--124, 1989.
\newblock \doi{10.1038/341119a0}.
\newblock URL \url{https://doi.org/10.1038/341119a0}.

\bibitem[Leff and Rex(2002)]{leff2002}
Harvey~S. Leff and Andrew~F. Rex, editors.
\newblock \emph{Maxwell's Demon 2: Entropy, Classical and Quantum Information,
  Computing}.
\newblock Institute of Physics Publishing, Bristol and Philadelphia, 2002.
\newblock ISBN 978-0750307596.

\bibitem[Sagawa and Ueda(2012)]{sagawa2012}
Takahiro Sagawa and Masahito Ueda.
\newblock Nonequilibrium thermodynamics of feedback control.
\newblock \emph{Physical Review E}, 85\penalty0 (2):\penalty0 021104, 2012.
\newblock \doi{10.1103/PhysRevE.85.021104}.

\bibitem[Sagawa(2014)]{sagawa2014}
Takahiro Sagawa.
\newblock Thermodynamic and logical reversibilities revisited.
\newblock \emph{Journal of Statistical Mechanics: Theory and Experiment},
  \penalty0 (3):\penalty0 P03025, 2014.
\newblock \doi{10.1088/1742-5468/2014/03/P03025}.

\bibitem[Parrondo et~al.(2015)Parrondo, Horowitz, and Sagawa]{parrondo2015}
Juan M.~R. Parrondo, Jordan~M. Horowitz, and Takahiro Sagawa.
\newblock Thermodynamics of information.
\newblock \emph{Nature Physics}, 11\penalty0 (2):\penalty0 131--139, 2015.
\newblock \doi{10.1038/nphys3230}.
\newblock URL \url{https://doi.org/10.1038/nphys3230}.

\bibitem[Ito and Sagawa(2015)]{ito2015}
Sosuke Ito and Takahiro Sagawa.
\newblock Maxwell's demon in biochemical signal transduction with feedback
  loop.
\newblock \emph{Nature Communications}, 6:\penalty0 7498, 2015.
\newblock \doi{10.1038/ncomms8498}.

\bibitem[Ito(2018)]{ito2018}
Sosuke Ito.
\newblock {Stochastic Thermodynamic Interpretation of Information Geometry}.
\newblock \emph{Physical Review Letters}, 121\penalty0 (3):\penalty0 030605,
  2018.
\newblock \doi{10.1103/PhysRevLett.121.030605}.

\bibitem[Levin(1974)]{levin1974}
Leonid~A. Levin.
\newblock {Laws of Information Conservation (Non-growth) and Aspects of the
  Foundation of Probability Theory}.
\newblock In \emph{Problems in the Transmission of Information}, pages
  206--210. Springer, 1974.
\newblock Early work related to algorithmic thermodynamics.

\bibitem[Chaitin(1975)]{chaitin1975}
Gregory~J. Chaitin.
\newblock A theory of program size formally identical to information theory.
\newblock \emph{Journal of the ACM}, 22\penalty0 (3):\penalty0 329--340, 1975.
\newblock \doi{10.1145/321892.321894}.
\newblock URL \url{http://www.cs.auckland.ac.nz/~chaitin/acm75.pdf}.

\bibitem[Chaitin(1987)]{chaitin1987}
Gregory~J. Chaitin.
\newblock \emph{Algorithmic Information Theory}, volume~1 of \emph{Cambridge
  Tracts in Theoretical Computer Science}.
\newblock Cambridge University Press, Cambridge, UK, 1987.
\newblock ISBN 978-0-521-34676-5.

\bibitem[Staiger(1998)]{staiger1998}
Ludwig Staiger.
\newblock A tight upper bound on kolmogorov complexity and uniformly optimal
  prediction.
\newblock \emph{Theory of Computing Systems}, 31\penalty0 (3):\penalty0
  215--229, June 1998.
\newblock \doi{10.1007/s002240000086}.
\newblock URL \url{https://link.springer.com/article/10.1007/s002240000086}.
\newblock Special issue on Computability of the Physical.

\bibitem[Calude et~al.(2004)Calude, Staiger, and Terwijn]{calude2004partial}
Cristian~S. Calude, Ludwig Staiger, and Sebastiaan~A. Terwijn.
\newblock On partial randomness.
\newblock Technical Report 239, Centre for Discrete Mathematics and Theoretical
  Computer Science, University of Auckland, December 2004.
\newblock URL
  \url{https://www.cs.auckland.ac.nz/research/groups/CDMTCS/researchreports/239cris.pdf}.
\newblock December~6, 2004.

\bibitem[Tadaki(2002)]{tadaki2002}
Kohtaro Tadaki.
\newblock A statistical mechanical interpretation of algorithmic information
  theory.
\newblock In \emph{Proceedings of the 4th International Conference on
  Unconventional Models of Computation (UMC'02)}, volume 2509 of \emph{Lecture
  Notes in Computer Science}, pages 242--251. Springer, 2002.
\newblock \doi{10.1007/3-540-36377-7_21}.

\bibitem[Tadaki(2008)]{tadaki2008}
Kohtaro Tadaki.
\newblock A statistical mechanical interpretation of algorithmic information
  theory.
\newblock \emph{CoRR}, abs/0801.4194, 2008.
\newblock URL \url{http://arxiv.org/abs/0801.4194}.

\bibitem[Baez and Stay(2012)]{baez2012}
John~C. Baez and Mike Stay.
\newblock Algorithmic thermodynamics.
\newblock \emph{Mathematical Structures in Computer Science}, 22\penalty0
  (5):\penalty0 771--787, 2012.
\newblock \doi{10.1017/S0960129511000520}.
\newblock Special Issue: Computability of the Physical.

\bibitem[Tadaki(2019)]{tadaki2019}
Kohtaro Tadaki.
\newblock \emph{A Statistical Mechanical Interpretation of Algorithmic
  Information Theory}, volume~36 of \emph{SpringerBriefs in Mathematical
  Physics}.
\newblock Springer Singapore, 2019.
\newblock ISBN 978-981-15-0739-7.
\newblock \doi{10.1007/978-981-15-0739-7}.
\newblock URL \url{https://doi.org/10.1007/978-981-15-0739-7}.

\bibitem[Ebtekar and Hutter(2025)]{ebtekar2025}
Aram Ebtekar and Marcus Hutter.
\newblock Foundations of algorithmic thermodynamics.
\newblock \emph{Phys. Rev. E}, 111:\penalty0 014118, Jan 2025.
\newblock \doi{10.1103/PhysRevE.111.014118}.
\newblock URL \url{https://link.aps.org/doi/10.1103/PhysRevE.111.014118}.

\bibitem[Jarzynski(1997)]{jarzynski1997}
C.~Jarzynski.
\newblock {Nonequilibrium Equality for Free Energy Differences}.
\newblock \emph{Physical Review Letters}, 78\penalty0 (14):\penalty0
  2690--2693, 1997.
\newblock \doi{10.1103/PhysRevLett.78.2690}.

\bibitem[Crooks(1999)]{crooks1999}
Gavin~E. Crooks.
\newblock {Entropy production fluctuation theorem and the nonequilibrium work
  relation for free energy differences}.
\newblock \emph{Physical Review E}, 60\penalty0 (3):\penalty0 2721--2726, 1999.
\newblock \doi{10.1103/PhysRevE.60.2721}.

\bibitem[Jarzynski(2011)]{jarzynski2011}
Christopher Jarzynski.
\newblock {Equalities and Inequalities: Irreversibility and the Second Law of
  Thermodynamics at the Nanoscale}.
\newblock \emph{Annual Review of Condensed Matter Physics}, 2:\penalty0
  329--351, 2011.
\newblock \doi{10.1146/annurev-conmatphys-062910-140506}.

\bibitem[Goodman(1972)]{goodman1972}
Nelson Goodman.
\newblock \emph{Fact, Fiction, and Forecast}.
\newblock Bobbs-Merrill, Indianapolis, 3rd edition, 1972.

\bibitem[Sneath and Sokal(1973)]{sneath1973}
Peter H.~A. Sneath and Robert~R. Sokal.
\newblock \emph{Numerical Taxonomy: The Principles and Practice of Numerical
  Classification}.
\newblock W. H. Freeman, San Francisco, 1973.
\newblock ISBN 0716706970.

\bibitem[Solomonoff(1964{\natexlab{a}})]{solomonoff1964a}
Ray~J. Solomonoff.
\newblock {A Formal Theory of Inductive Inference. Part I}.
\newblock \emph{Information and Control}, 7\penalty0 (1):\penalty0 1--22,
  1964{\natexlab{a}}.
\newblock \doi{10.1016/S0019-9958(64)90223-2}.

\bibitem[Solomonoff(1964{\natexlab{b}})]{solomonoff1964b}
Ray~J. Solomonoff.
\newblock {A Formal Theory of Inductive Inference. Part II}.
\newblock \emph{Information and Control}, 7\penalty0 (2):\penalty0 224--254,
  1964{\natexlab{b}}.
\newblock \doi{10.1016/S0019-9958(64)90131-7}.

\bibitem[Hutter(2005)]{hutter2005}
Marcus Hutter.
\newblock \emph{Universal Artificial Intelligence: Sequential Decisions Based
  on Algorithmic Probability}.
\newblock Texts in Theoretical Computer Science. An {EATCS} Series. Springer,
  2005.
\newblock ISBN 978-3-540-22139-5.
\newblock \doi{10.1007/b138233}.

\bibitem[Pathria and Beale(2011)]{pathria2011}
R.~K. Pathria and Paul~D. Beale.
\newblock \emph{Statistical Mechanics}.
\newblock Academic Press, 3rd edition, 2011.
\newblock ISBN 978-0123821881.

\bibitem[Huang(1987)]{huang1987}
Kerson Huang.
\newblock \emph{Statistical Mechanics}.
\newblock Wiley, 2nd edition, 1987.
\newblock ISBN 978-0471815181.

\bibitem[Callen(1985)]{callen1985}
Herbert~B. Callen.
\newblock \emph{Thermodynamics and an Introduction to Thermostatistics}.
\newblock Wiley, 2nd edition, 1985.
\newblock ISBN 978-0471862567.

\bibitem[Cilibrasi and Vitanyi(2007)]{cilibrasi2007}
Rudi~L. Cilibrasi and Paul~M.B. Vitanyi.
\newblock {The Google Similarity Distance}.
\newblock \emph{IEEE Transactions on Knowledge and Data Engineering},
  19\penalty0 (3):\penalty0 370--383, 2007.
\newblock \doi{10.1109/TKDE.2007.48}.

\bibitem[Li et~al.(2004)Li, Chen, Li, Ma, and Vitanyi]{li2004}
Ming Li, Xin Chen, Xin Li, Bin Ma, and P.M.B. Vitanyi.
\newblock The similarity metric.
\newblock \emph{IEEE Transactions on Information Theory}, 50\penalty0
  (12):\penalty0 3250--3264, 2004.
\newblock \doi{10.1109/TIT.2004.838101}.

\end{thebibliography}

\end{document}